\documentclass[sigconf, balance=false]{acmart}

\AtBeginDocument{%
  \providecommand\BibTeX{{%
    \normalfont B\kern-0.5em{\scshape i\kern-0.25em b}\kern-0.8em\TeX}}}

\copyrightyear{2020}
\acmYear{2020}
\setcopyright{acmcopyright}
\acmConference[AIES '20]{Proceedings of the 2020 AAAI/ACM Conference on AI, Ethics, and Society}{February 7--8, 2020}{New York, NY, USA}
\acmBooktitle{Proceedings of the 2020 AAAI/ACM Conference on AI, Ethics, and Society (AIES '20), February 7--8, 2020, New York, NY, USA}
\acmPrice{15.00}
\acmDOI{10.1145/3375627.3375823}
\acmISBN{978-1-4503-7110-0/20/02}

\usepackage{amsmath}
\usepackage{flushend}
\DeclareMathOperator*{\argmax}{arg\,max}
\usepackage{tikz-cd}
\usetikzlibrary{positioning}
\usetikzlibrary{calc}
\usetikzlibrary{arrows}
\usetikzlibrary{decorations.markings}
\usetikzlibrary{arrows.meta}
\usetikzlibrary{bayesnet}
\usepackage{float}
\usepackage{bbold}

\tikzset{strike thru arrow/.style={
    decoration={markings, mark=at position 0.5 with {
        \draw [blue, thick,-] 
            ++ (-0.15cm, -0.15cm) 
            -- ( 0.15cm,  0.15cm);}
    },
    postaction={decorate},
}}

\theoremstyle{plain}
\newtheorem{thm}{Theorem}

\newtheorem{lem}[thm]{Lemma}

\theoremstyle{definition}
\newtheorem{defn}[thm]{Definition}

\theoremstyle{remark}

\newcommand{\nc}{\newcommand}
\nc{\dmo}{\DeclareMathOperator}
\nc{\rnc}{\renewcommand}

\nc{\B}[1]{\mathbb{#1}}  
\renewcommand{\C}[1]{\mathcal{#1}}
\nc{\Sc}[1]{\mathscr{#1}}
\nc{\F}[1]{\mathfrak{#1}}
\nc{\cdes}[1][]{\cde_{\Ob #1}}
\nc{\cdeshat}[1][]{\widehat{\cde}_{\Ob #1}}
\nc{\tes}[1][]{\te_{\Ob #1}}
\nc{\bbE}{\B{E}}  

\nc{\wt}{\widetilde}
\nc{\wh}{\widehat}
\nc{\la}{\langle}
\nc{\ra}{\rangle}
\nc{\fvert}{\left\vert\vphantom{\frac11}\right.}
\nc{\unif}{\textsc{Unif}}
\nc{\cloan}{c_{\textsc{Loan}}}
\nc{\cinfo}{c_{\textsc{Info}}}
\nc{\bern}{\textsc{Bern}}
\nc{\decmak}{decision-maker }

\nc{\indep}{\protect\mathpalette{\protect\independenT}{\perp}}
\def\independenT#1#2{\mathrel{\rlap{$#1#2$}\mkern2mu{#1#2}}}

\rnc{\epsilon}{\varepsilon}

\dmo{\cde}{\textsc{cde}}
\dmo{\te}{\textsc{te}}
\dmo{\var}{\textsc{Var}}
\dmo{\bnm}{\textsc{Binom}}
\dmo{\EE}{\mathbb{E}}
\dmo{\Ob}{Ob}
\dmo{\naive}{\textsc{naive}}
\nc{\cs}{c_{\textsc{S}}}
\nc{\ca}{c_{\textsc{A}}}
\nc{\varcs}{C^{\textsc{S}}}

\dmo{\as}{\mathbf{AS}}
\dmo{\rs}{\mathbf{RS}}
\dmo{\ns}{\mathbf{NS}}
\dmo{\rrs}{\mathbf{rRS}}

\begin{document}
\fancyhead{}
\title{Fair Allocation through Selective Information Acquisition}

\author{William Cai}
\orcid{1234-5678-9012}
\affiliation{%
  \institution{Stanford University}
  \postcode{94305}
}
\email{willcai@stanford.edu}

\author{Johann Gaebler}
\orcid{1234-5678-9012}
\affiliation{%
  \institution{Stanford University}
  \postcode{94305}
}
\email{jgaeb@stanford.edu}

\author{Nikhil Garg}
\orcid{1234-5678-9012}
\affiliation{%
  \institution{Stanford University}
  \postcode{94305}
}
\email{nkgarg@stanford.edu}

\author{Sharad Goel}
\orcid{1234-5678-9012}
\affiliation{%
  \institution{Stanford University}
  \postcode{94305}
}
\email{scgoel@stanford.edu}

\renewcommand{\shortauthors}{Cai et al.}

\begin{abstract}
Public and private institutions must often
allocate scarce resources under uncertainty. 
Banks, for example, extend credit to loan applicants based in part on their estimated likelihood of repaying a loan.
But when the quality of information differs across candidates (e.g., if some applicants lack traditional credit histories), common lending strategies can lead to disparities across groups.
Here we consider a setting in which decision makers---before allocating resources---can choose to spend some of their limited budget further screening select individuals.
We present a computationally efficient algorithm for deciding whom to screen that maximizes a standard measure of social welfare. 
Intuitively, decision makers should screen candidates on the margin, for whom the additional information could plausibly alter the allocation. 
We formalize this idea by showing the problem can be reduced to solving a series of linear programs.
Both on synthetic and real-world datasets, this strategy improves utility, illustrating the value of targeted information acquisition in such decisions.
Further, when there is social value for distributing resources to groups for whom we have \emph{a priori} poor information---like those without credit scores---our approach can substantially improve the allocation of limited assets. 
\end{abstract}




\maketitle

\section{Introduction}
Approximately one in seven U.S. households have unmet demand for small-dollar loans,
and are often unable to secure credit from traditional financial institutions as they have little or no formal credit history~\cite{fdic2018}.
However, in the majority of these households, individuals receive regular income and typically pay their bills on time, 
which suggests many in fact would have low risk of default~\cite{fdic2018}.
One barrier to providing loans to this low-risk yet underserved subpopulation is that it can be more expensive and time-consuming to screen individuals with non-traditional financial histories,
limiting the inclusiveness of the banking system.

Motivated by this problem, we propose and analyze a
strategy in which one can pay to acquire additional information on applicants. 
The task for a budget-constrained decision maker is thus to first select a set of candidates to screen and then, given the results of that screening process, determine to whom to allocate the remaining resources.
In modeling this situation, we assume there is a fixed cost for screening each applicant,
and that decision makers have prior knowledge of the distribution of information they would receive if they choose to screen an applicant.
In practice, we note that such prior knowledge could be obtained by screening a small random sample of applicants to learn the resulting information distributions.

We derive an efficient algorithm for computing an optimal, utility-maximizing strategy for the general screening and allocation problem.
To do so, we first show that once a set of candidates has been selected to screen, it is optimal to allocate the remaining resources according to a threshold rule, with
assets distributed to those candidates having post-screening expected utility above a fixed threshold.
Further, for any fixed threshold policy, we show that one can find the optimal set of candidates to screen (while satisfying the budget constraint) via a linear program.
Intuitively, one should screen candidates near the margin, for whom the screening process could reveal information that could push a candidate across the threshold.
But to do this rigorously, one also needs to account for the precise structure of the prior information.
Finally, we sweep over the possible thresholds, solving the corresponding linear program at each point.
In this manner, we obtain both a rule to screen candidates and a specific threshold policy for distributing funds to candidates with sufficiently high post-screening value.

We further consider an extension of the above problem in which policymakers have explicit value for diversity. 
For example, instead of simply finding a max-utility policy, one might maximize utility subject to the constraint that a particular group---such as those who traditionally have had limited access to credit markets---achieve at least a fixed minimum utility.
In the United States, those with unmet demand for credit are disproportionately black and Hispanic~\cite{fdic2018}, heightening the value of diversity considerations in allocation decisions.
We show that this (and related) extensions can be incorporated into our general algorithmic approach in a straightforward manner.

To demonstrate the potential value of augmenting allocation decisions with a screening phase, 
we apply our methods to both synthetic datasets and one with real measures
 of creditworthiness.
In particular, we examine the potential benefits of screening as a function of the cost and value of information.
Especially when we impose a diversity constraint,
we find that screening strategies can significantly outperform a naive strategy that simply attempts to satisfy the constraint without screening any applicants.

For concreteness, we frame our discussion in terms of lending decisions, though our approach applies to many allocation settings.
For example, it is often challenging to accurately assess household wealth---particularly in countries where informal and irregular work is more common---and, in turn, to appropriately target the distribution of government subsidies~\cite{noriega2019active}. 
The simple strategy of distributing funds to those families clearly in need can systematically overlook populations with harder-to-verify financial status.
As with lending, one can judiciously allocate some of the budget to more extensively screen certain applicants, ensuring funds are ultimately distributed to those who can benefit the most.

\section{Related Work}

Several papers address the problem of \emph{active feature acquisition}~\citep{melville2004active,melville2005expected,saar2009active},
where one can selectively purchase missing data to improve the overall out-of-sample performance of a statistical model.
We consider the related problem of acquiring features to identify specific, high-value individuals.
While there is some shared intuition between the two settings---that one should seek information on individuals most likely to alter downstream decisions---the technical approach we take is different, in large part because our end goal is optimal allocation rather than statistical learning.

In a related, recent stream of research,
\citet{bakker2019fairness} and \citet{noriega2019active}
likewise consider a feature acquisition problem, but with
fairness constraints.
In their setting, the decision maker must acquire additional features for each individual to ensure classification decisions have similar errors rates across groups---including parity in false negative and false positive rates---a common measure of fairness in the machine learning community~\cite{hardt2016,kleinberg2016inherent,chouldechova2016fair}.
Our approach to the problem differs in three important respects.
First, we adopt the perspective of constrained utility maximization.
Past work has shown that directly equalizing error rates can lead to outcomes that, counterintuitively, may harm the very groups they were designed to protect~\cite{corbett2018measure,corbett2017,liu2018delayed}.
We avoid such deleterious outcomes by instead framing the problem explicitly in terms of group-specific utilities:
fairness is encoded into our requirement that the decision maker must allocate some minimum amount of utility to each group.
Second, we focus on one-shot screening decisions, in which decision makers simply choose whether or not to acquire information on each individual, rather than sequentially deciding how much information to acquire based on the results of each past acquisition decision.
Our one-shot formulation maps to the binary decision structure (i.e., to screen or not to screen) common in many institutions and leads to different optimization challenges.
Third, we directly model the tradeoff between screening and allocation decisions by tying both to a common budget constraint (i.e., more screening means less funds are available to ultimately distribute to individuals).

\citet{elzayn2019fair} also consider equitable ways to allocate resources, but in a setting where learning occurs by repeatedly allocating resources instead of by purchasing information. In comparison to our work, they focus on solving the problem of censored feedback, where the learner may not understand, and thus never allocate resources to, groups who did not previously receive resources.

Finally, our work touches on research from the fair division and allocation literature \citep{brams1996fair,
moulin2004fair,thomson2011fair}, which considers how to share resources while satisfying fairness properties defined between individuals.
The former devises mechanisms wherein strategic agents self-divide the resources fairly, and the latter studies the existence of allocations that jointly satisfy various fairness notions. 
In contrast to our work, that line of research is particularly concerned with individual incentives, strategic action, and equilibrium effects.

\section{A Model of Screening and Allocation}

We model screening and allocation decisions as a sequential process in which a budget-constrained lender first selects a (possibly random) subset of candidates to further screen from a pool of \(n\) applicants, and then, based on the information revealed in that screening phase, selects a second (possibly random) subset of candidates to receive a loan.

We assume the value of lending to an applicant \(i\) is given by the random variable \(U_i\).
These utilities are intended to capture the full social value of providing loans, and we imagine the lender aims to optimize social welfare, as in the case of a government agency.
In general, the lender has only partial information about \(U_i\). More specifically, if the lender chooses not to screen an applicant, we assume the lender knows only the applicant's conditional expectation \(\mu_i = \EE [U_i \mid X_i = x_i]\)
given their pre-screening covariates \(x_i\), such as credit score for those applicants who have traditional credit histories.
On the other hand, if the lender opts to screen an applicant, they learn \(\EE[U_i \mid X_i = x_i, \tilde{X}_i = \tilde{x}_i]\), where \(\tilde{x}_i\) denotes the additional information one gains through screening. 
For example, \(\tilde{x}_i\) might encode applicant \(i\)'s history of paying
their electricity or phone bills---information that is often feasible to acquire with some extra effort and which is a good indicator of creditworthiness~\cite{fdic2018}.

When deciding whom to screen, we assume 
the lender knows the distribution of \(\C D = (D_1, \ldots, D_n)\), where \(D_i = \EE[U_i \mid X_i = x_i, \tilde{X}_i]\).
That is, the lender knows how their estimate of utility could change if they decide to screen each applicant, where these distributions may depend on the available pre-screening covariates.
A lender may, for example, thus choose only to screen applicants whose estimate is likely to substantially change given additional information.
We further assume the lender must pay a
fixed cost \(\cs\) for screening an applicant
and a cost \(\ca\) for underwriting a loan. 
For simplicity we assume these costs do not vary across applicants, though it is straightforward to extend to the more general case; see the online appendix for details.

Based on knowledge of the above information and cost structure, the lender selects a randomized strategy to screen applicants.
That is, the lender chooses a vector 
\(\C P = (p_1, \dots, p_n)\), meaning that each applicant \(i\) is selected to be screened independently with probability \(p_i\).
Let \(\C S = (S_1, \dots, S_n)\) indicate which applicants are ultimately screened under this policy; therefore, \(S_i \in \{0,1\}\) is a Bernoulli random variable with probability of success \(p_i\).

Given this randomized screening policy, we can now write the information \(\hat {\C U} = (\hat U_1, \ldots, \hat U_n)\) the lender has at the end of the screening phase as follows:
    \begin{equation}
        \hat U_i = \begin{cases}
            \EE[U_i \mid X_i = x_i] & \text{ if } S_i = 0, \\
            \EE[U_i \mid X_i = x_i, \tilde{X}_i] & \text{ if } S_i =1.
        \end{cases}
    \end{equation}
In other words, \(\hat U_i\) is the lender's post-screening estimated utility of giving a loan to applicant \(i\).
In particular, if \(S_i = 1\) (i.e., the applicant is screened), the lender's estimate changes from \(\EE[U_i\mid X_i = x_i]\),
the estimate based only on applicant \(i\)'s pre-screening covariates \(x_i\),
to 
\(\bbE[U_i \mid X_i = x_i, \tilde{X}_i = \tilde{x}_i] \), which incorporates the post-screening information \(\tilde{x}_i\).

Finally, the lender chooses an allocation policy, denoted \(\C A(\hat {\C U}) = (A_1, \dots A_n)\), 
where \(A_i \in [0,1]\)
specifies the probability a loan is (independently) offered to each applicant.
Importantly, \(\C A\) is a function of 
the lender's post-screening utility estimates
\(\hat {\C U}\).
For example, the lender might give loans to the individuals with the highest post-screening utility estimates, up to the budget constraint.

Combining all of the above, the lender's optimization problem is to choose screening and allocation policies \((\C P^*, \C A^*)\) that maximize expected welfare,
\begin{equation}
\label{eq:def1}
    (\C P^*, \C A^*) \in \argmax_{\C P, \C A} \bbE\left[\sum_{i = 1}^n U_i A_i\right],
\end{equation}
subject to being budget-balanced in expectation,%
\footnote{By requiring the budget constraint to hold only in expectation---rather than exactly---we are able to find a computationally tractable solution to the problem. In practice, such a constraint means that agencies are allowed some flexibility as long as they don't overspend on average,
which, we believe, is often a realistic requirement.
}

\begin{equation}
\label{eq:def2}
    \EE \left[ \sum_{i = 1}^n \cs S_i + \ca  A_i \right] \leq B,
\end{equation}
where \(B\) is a fixed, non-negative constant.

In some settings, decision makers may value diversity in their allocations.
For example, they may wish to ensure a certain minimum number of loans are provided to groups that historically have been excluded from credit markets.
One can encode this policy preference directly into the utilities, in which case the resulting optimization problem would incorporate one's value for diversity.
That approach, however, requires decision makers to agree upon these utilities to interpret the results, which can be challenging. 

Here we take a complementary approach that explicitly allows value for diversity to differ across decision makers.
Suppose the applicant pool is partitioned into \(m\) groups \(\{G_1, \ldots, G_m\}\);
for example, if \(m = 2\), we might partition candidates into those who traditionally have had access to credit markets and those who have not.
Then we require the selected policy \((\C{P}^*, \C{A}^*)\)
to allocate at least \(\Lambda_j \geq 0\) utility to group \(G_j\), where these utilities do not themselves include any value for diversity:
    \begin{equation}
    \label{eq:def3}
        \EE \left[ \sum_{i \in G_j} U_i A_i \right] \geq \Lambda_j.
    \end{equation}
In practice, as we discuss below, one would solve this optimization problem for a range of \(\Lambda\), which traces out the Pareto frontier of possible policies across different group constraints, corresponding to different values for diversity.
The diversity condition above is expressed in terms of utility, but we might, alternatively,
simply lower bound the \emph{number} of loans \(\sum_{i \in G_j} A_i \) given to members of each group \(G_j\).
This alternative constraint can be handled in a straightforward manner by our algorithm detailed below.

\subsection{A stylized example}
We illustrate the above ideas in the context of a simple, stylized example.
Suppose a lender must decide how best to allocate loans among an applicant pool of 13 people, with an overall budget of \(\$2{,}000\). 
Providing a loan costs \(\ca = \$400\), 
and additional screening of an applicant costs \(\cs = \$50\).

Further suppose that five of the applicants are able to provide rich credit histories, and the expected utility of giving each of them a loan is \(\mu_i = \$750\). For this group, additional screening would not provide any more information.
Imagine that the other eight applicants do not have formal credit histories, and the utility of giving them a loan is accordingly lower due to the risk of default, with \(\mu_i = \$500\).
However, the lender knows that these applicants
come in two types that could be disambiguated through additional screening.
More specifically, the lender knows that after screening, individuals in this group can be divided into those with expected utility \$1,000 (with 50\% chance) and those with expected utility \$0 (with 50\% chance). 
The high-utility group could, for example, correspond to those who demonstrate a history of consistently paying their bills on time.

The naive strategy that does not screen any individuals would allocate five loans to hit the budget constraint, since \(5 \times \$400 = \$2{,}000\).
All five loans would go to applicants with rich credit histories (\(\mu_i = \$750\)) over those without (\(\mu_i = \$500\)).
In this case, the total expected utility of the no-screening allocation is \(5 \times \$750 = \$3{,}750\).

But in this scenario, one can improve overall utility by screening all eight applicants without formal credit histories and then granting loans to the high-utility applicants that are identified.
Under that strategy, we expect four of the eight screened applicants to be identified as high utility (\(\$1{,}000\)),
and so the expected utility of the allocation is \(4 \times \$1{,}000 = \$4{,}000\), greater than the expected utility of \$3,750 under the no-screening strategy.
Finally, the expected cost of the strategy is 
\(8 \times \$50 = \$400\) for screening plus 
\(4 \times \$400 = \$1{,}600\) for distributing the loans, totaling \(\$2{,}000\) and satisfying the budget constraint.
In this example, one can thus improve overall utility---and even allocate more loans to the group that \emph{a priori} appears worse---by incorporating additional screening into the decision-making process.

\section{Finding Optimal Policies}
\label{sec:solution}

We now derive an efficient algorithm  to find optimal screening and allocation policies \((\C{P}^*, \C{A}^*)\), subject to the budget and  diversity constraints. 
We start by showing that 
over the full space of policies,
it is optimal to allocate resources according to a
\emph{threshold policy}, with loans
dispersed to individuals having post-screening expected utility \(\hat U_i\) above a group-specific threshold \(t_j\) for \(j = 1, \ldots, m\).
Then, for each threshold policy, we show the optimal screening policy can be obtained by solving a linear program (LP), a type of optimization problem with linear objective function and linear constraints, for which their exist fast solution methods~\cite{bertsimas1997introduction}.
As a result, we can find an optimal combined screening and allocation policy by sweeping over threshold policies and solving the corresponding LP for each such policy.

We begin by formally defining threshold policies.
\begin{defn}[Threshold Policy]
    A \emph{threshold policy} is an allocation policy \(\C A\) for some fixed \(t_j \in \B R \cup \{-\infty, \infty\}\) and \(\alpha_j \in [0,1]\), \(j = 1, \ldots, m\), such that
        \begin{equation*}
            A_i = \begin{cases}
                1 & \hat U_i > t_{g_i},\\
                \alpha_{g_i} & \hat U_i = t_{g_i}, \\
                0 & \hat U_i < t_{g_i},
            \end{cases}
        \end{equation*}
    where \(g_i\) denotes the group membership of individual \(i\).
\end{defn}

Threshold policies deterministically allocate resources to those with post-screening expected utilities above a
fixed, group-specific threshold \(t_j\). Randomization (i.e., allocating resources with probability \(\alpha_j\)) at the threshold \(t_j\) may be necessary to exactly satisfy the budget constraint,
which is important because for an individual not screened, the distribution of \(\hat U_i\) is concentrated at a single point.

Theorem~\ref{thm:main} below formally states that
it is sufficient to restrict  to the set of threshold policies when searching for a globally optimal screening and allocation policy.

\begin{thm}
\label{thm:main}
    Suppose the constrained optimization problem defined by Eqs.~\eqref{eq:def1}, \eqref{eq:def2}, and \eqref{eq:def3} has a solution \((\C P^*, \C A^*)\). Then there is a threshold policy \(\C T^*\) such that \((\C P^*, \C T^*)\) is also a solution.
\end{thm}

To see this, suppose that \(\C P^*\) deterministically selects a subset of applicants to screen.
Then, for each group \(G_j\), it is clear one should allocate loans to the approximately \(k_j\) individuals in each group with the highest post-screening estimated utility, where \(k_j\) is the number of loans granted to each group under \(\C A^*\).
Such a rank-based allocation can 
equivalently be written as a threshold rule
with the same expected utility.
The more general case, in which screening decisions are randomized, introduces some technical complications, but the spirit of the argument is similar.
Full details can be found in the online appendix.

Now, 
given a threshold policy \(\C T\) with  thresholds \(t_1, \dots, t_m  \in \B R\) and boundary randomization probabilities \(\alpha_1, \dots, \alpha_m \in [0,1]\),
we turn to finding an optimal companion screening policy. 
We show that such an optimal screening policy \(\C P\) can be found by solving an LP.
In particular, the LP has \(n\) decision variables \(p_1, \dots, p_n\), with \(p_i \in [0,1]\) specifying the probability that applicant \(i\) is screened;
the objective equals the utility of the 
combined screening and allocation policy;
and the constraints encode our budget and diversity conditions.

To construct the LP, we first define the following quantities that depend on both the threshold rule \(\C T\) and the lender's prior knowledge on the value of screening:
    \begin{align*}
        q_i &= \Pr(D_i > t_{g_i}),\\
        e_i &= \EE[U_i \mid D_i > t_{g_i}],\\
        o_i &= \B 1_{\mu_i > t_{g_i}} + \alpha_{g_i} \cdot \B 1_{\mu_i = t_{g_i}},
    \end{align*}
where \(g_i\) denotes applicant \(i\)'s group membership. Recall that 
\(\mu_i\) denotes pre-screening expected utility
and \(D_i\) denotes post-screening expected utility 
if applicant \(i\) is screened. 
Both \(\mu_i\) and the distribution of \(D_i\) are known in advance.

For fixed screening probabilities \(p = (p_1, \dots p_n)\), the expected utility of the corresponding 
screening and allocation policy
can now be expressed as a linear function of \(p\):
    \begin{equation}
    \label{eq:obj}
        \sum_{i = 1}^n \left[ q_i \cdot e_i \cdot p_i + o_i \cdot \mu_i \cdot (1 - p_i) \right].
    \end{equation}
The first summand reflects the expected utility associated with applicant \(i\) if they were screened,
and the second summand reflects the expected utility if they were not screened.

Our goal is to maximize the expression in Eq.~\eqref{eq:obj} subject to the budget and diversity conditions, which we now show can also be expressed as linear constraints on \(p\).
In terms of the constants defined above,
the budget constraint can be written as
\begin{equation}
    \label{eq:constr1}
        \sum_{i = 1}^n \left[ \cs \cdot p_i + \ca \cdot q_i \cdot p_i + \ca \cdot o_i \cdot (1-p_i) \right] \leq B.
    \end{equation}
Likewise, the \(j\) diversity constraints in Eq.~\eqref{eq:def3} become
    \begin{equation}
    \label{eq:constr2}
        \sum_{i \in G_j} \left[ q_i \cdot e_i \cdot p_i + o_i \cdot \mu_i \cdot (1 - p_i) \right] \geq \Lambda_j,
    \end{equation}
for \(j = 1, \ldots, m\).

Together, the objective 
given by Eq.~\eqref{eq:obj},
with constraints defined by Eqs.~\eqref{eq:constr1} and \eqref{eq:constr2},
define a linear program with decision variables $\{p_i\}$, the solution of which---if one exists--- gives a screening policy \(\C P^*\) that is optimal for the threshold policy \(\C T\).
A jointly optimal screening and allocation policy
\((\C P^*, \C T^*)\) can accordingly be found through a grid search
over all threshold policies \(\C T\) in a (discretized) space \(\B R^m \times [0, 1]^m\).
For each choice of policy \(\C T\), defined by threshold and randomization parameters \(t_1, \ldots, t_m\) and \(\alpha_1, \ldots, \alpha_m\), we find the corresponding optimal screening policy \(\C P\) by solving the LP described above.
Then, among the resulting screening-allocation pairs,
the one 
maximizing utility is guaranteed to be globally optimal.

\section{Experiments}
We now investigate the value of our sequential screening and allocation approach through two simulation exercises, one based on synthetic data and another based on real loan data.
First, with the synthetic data, we examine the optimal policies as we vary both the cost of screening and the value of the resulting information.
Then, with the real-world data, we illustrate how a decision maker could, in practice, operationalize our screening and allocation approach to pursue equity when distributing limited resources.

In each experiment, we allocate loans to members of two groups.  
One of the groups---which we call the ``targeted'' group---can be screened for more information at some cost, and we imagine there is social value to distributing more resources to this group.
For example, the targeted group may be comprised of those who traditionally have not had ready access to the banking system and accordingly have limited formal credit history.
For simplicity, we further assume that those in the non-targeted group cannot be screened, perhaps because they already have complete credit histories.

We define the utility of lending
to applicant \(i\) to be
    \begin{equation}
    \label{eq:exp-utility}
        U_i = \begin{cases}
            a & Y_i = 1,\\
            b & Y_i = 0,
        \end{cases}
    \end{equation}
where \(Y_i \in \{0,1\}\) indicates whether they would pay back the loan,
and \(a\) and \(b\) are fixed, known constants.
In both of our experiments, we set \(a = \$1{,}000\) and \(b=-\$200\),
meaning there is \(\$1{,}000\) of social utility when an applicant receives and pays back a loan and \(-\$200\) utility when an applicant defaults on a loan, for example because defaulting could trigger further financial distress. As above, we imagine the lender is a government institution or other agent attempting to maximize total social utility.
We further assume that the pre-screening covariate \(x_i \in [0,1]\) specifies the lender's pre-screening estimate of applicant \(i\)'s likelihood to repay a loan;
in other words, \(x_i = \Pr(Y_i = 1 \mid X_i = x_i)\).
The covariate \(x_i\) can thus be translated into (pre-screening) expected utility \(\mu_i\) by Eq.~\eqref{eq:exp-utility}:
    \begin{equation*}
        \EE [U_i \mid X_i = x_i] = a x_i + b(1 - x_i).
    \end{equation*}

\begin{figure}[t]
\centering
\includegraphics[width=6.5cm]{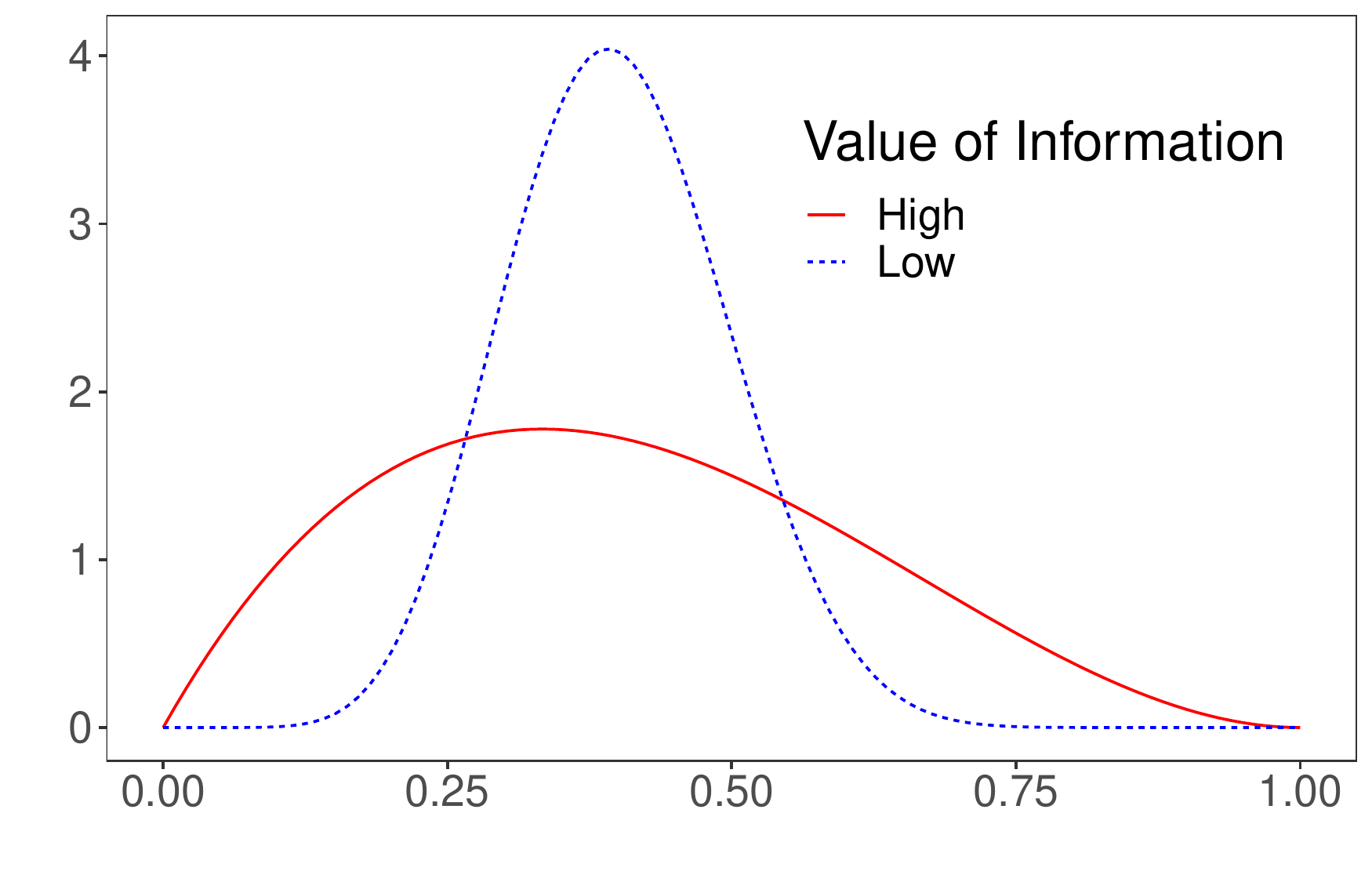}
\caption{For an applicant with pre-screening likelihood of being creditworthy \(x_i = 0.4\), the distribution of post-screening creditworthiness, 
\(\Pr[Y_i = 1 \mid X_i = 0.4, \tilde{X}_i ]\).  
Under a regime with high value of information (red line), the post-screening distribution has higher variance than in a setting with low value of information (blue line).}
\label{fig:beta}
\end{figure}

\subsection{Synthetic data}\label{sec:synthetic}
We illustrate the value of screening in four regimes of low vs.\ high cost of information paired with low vs.\ high value of information. 
To do so, we created four synthetic datasets, each comprised of \(n = 500\) individuals.

For all four datasets, we first evenly split the population into targeted and non-targeted groups. 
For each applicant in the targeted group, we generated their pre-screening probability of repayment \(x_i\)
(or, equivalently, their observed pre-screening covariates) by independently drawing from a beta distribution%
\footnote{We use the beta as it is a familiar, easily-parameterized function with support on $[0, 1]$.} 
with mean 0.5 and count parameter 50.\footnote{
    In terms of the \(\alpha\) and \(\beta\) shape parameters often used to parameterize beta distributions,
    the mean is \(\alpha/(\alpha + \beta)\) and the count parameter is \(\alpha + \beta\). A higher count corresponds to a lower variance.
}
Similarly, for each applicant in the non-targeted group, we generated their pre-screening probability of repayment
by independently drawing from a beta distribution with mean 0.70 and count parameter 50.
The higher mean repayment probability for members of the non-targeted group corresponds to them being, on average, more creditworthy.

Now, for each applicant in the targeted group, the lender may elect to screen them. 
As a result of screening, the lender receives an improved estimate \(\tilde{x}_i\) of the applicant's repayment probability, so that:
\begin{align*}
\EE [U_i \mid X_i = x_i, \tilde{X}_i = \tilde{x}_i] = a \tilde{x}_i + b(1 - \tilde{x}_i).
\end{align*}
We assume \(\tilde{x}_i\) is
drawn from a beta distribution with mean \(x_i\)
(i.e., the lender's pre-screening estimate of the applicant's repayment probability).
In the high-information scenario, we set the 
count parameter for this beta distribution to be 5; 
and in the low-information scenario, we set it equal to 25.

Figure~\ref{fig:beta} shows these two post-screening information distributions for an applicant with \(x_i = 0.4\).
As illustrated in the plot, the high-information distribution (red line) has higher variance than the low-information distribution (blue line), and so screening is more likely to reveal very high risk and very low risk applicants in the high-information setting.
Finally, we set the cost of screening \(\cs\)
to be \(\$25\) in the low-cost scenario and \(\$100\) in the high-cost scenario, the cost of a loan to be $\ca = \$1{,}000$, and the total budget to be \(B = \$50{,}000\).

\begin{figure}[t]
    \centering
    \includegraphics[width=\columnwidth]{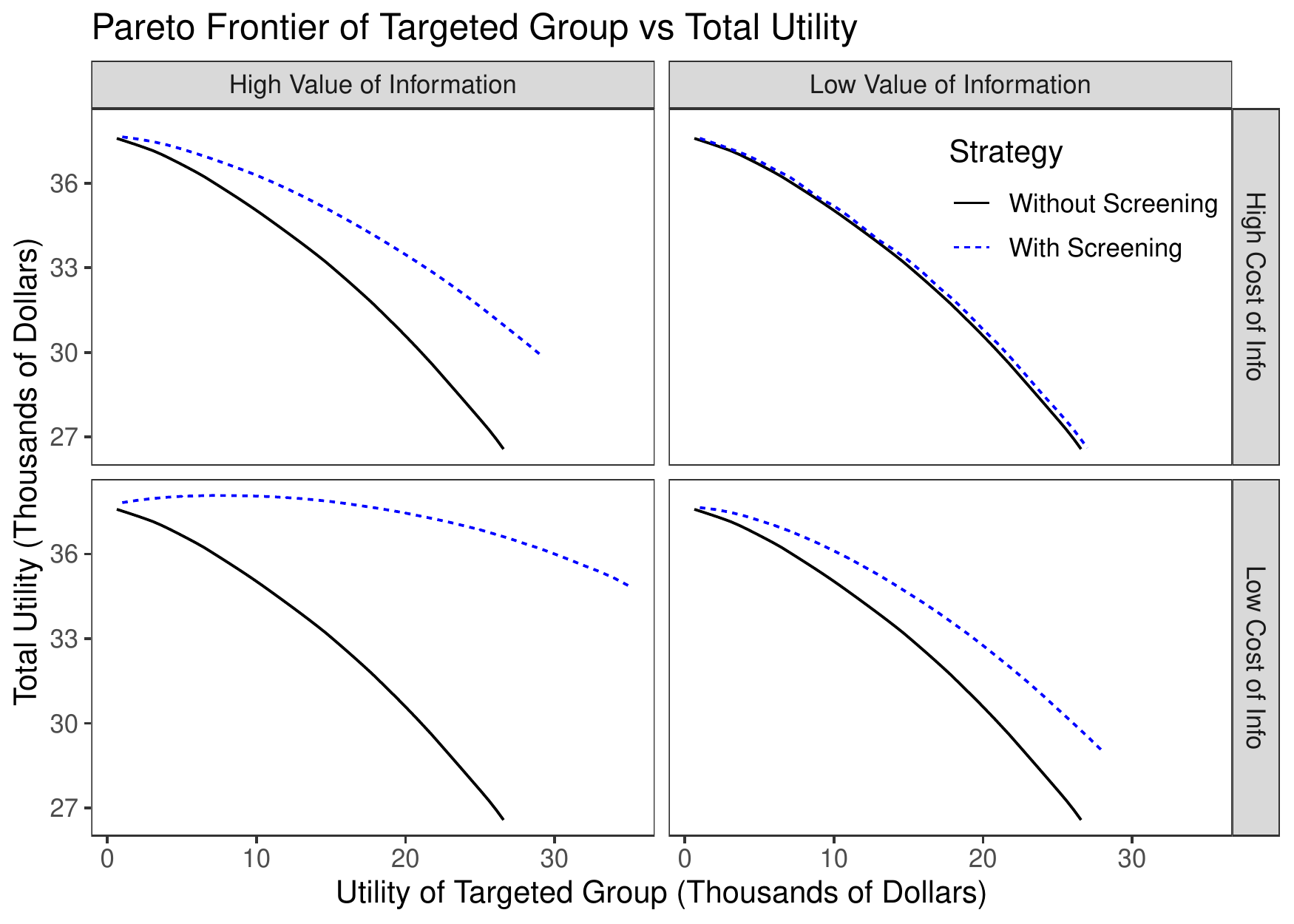}
    \caption{Comparison of our optimal screening strategy (blue line) with one without screening (black line) in four regimes with different costs and values of information.  
}
\label{fig:4b4}
\end{figure}

Given these four datasets, we computed the optimal screening and allocation strategies using the algorithm described above.\footnote{Because only one group can be screened in our experiments, we use a faster variant of our optimization algorithm, described in
the online appendix.
For any fixed set of parameters, this approach returns the optimal screening and allocation policy within a few seconds with the open-source LP solver SCS~\cite{o2016conic}.}
In our original problem formulation, the diversity constraint in Eq.~\eqref{eq:def3} specified only that we lower bound the utility of each group.
To better understand the impact of diversity on utility, we modify this constraint to be a strict equality for the utility of the targeted group and set the lower bound on utility to be \(\$0\) for the non-targeted group.
Thus, across a range of exactly satisfied utilities for the targeted group, we find the strategy that maximizes overall utility.

Figure~\ref{fig:4b4} shows the results of this analysis, with the blue lines tracing out the Pareto frontiers for each of the four scenarios we consider.
For comparison, the black lines show the corresponding result under a strategy that does not screen any applicants. 
Specifically, for any fixed utility constraint on the targeted group, the optimal no-screening policy first allocates loans to the $k$ individuals in the targeted group most likely to repay based on their pre-screening estimates $x_i$, where $k$ is chosen to satisfy the utility constraint; 
and then any remaining budget is used to allocate loans to those in the non-targeted group most likely to repay.

When either the value of information is high (left column) or the cost of screening is low (bottom row), we find that screening can be a valuable tool to improve utility.
Notably, screening plays a more important role in these examples as we demand greater utility be allocated to the targeted group,
since the no-screening strategy ends up dispersing loans to relatively high-risk applicants in the targeted group even though more creditworthy applicants in that group could be identified for little marginal cost.
As one might expect, the gap between the screening and no-screening strategies is particularly large when both information is valuable and cheap.
Indeed, in the high-value, low-cost setting (lower-left panel), one can achieve substantial diversity with little drop in overall utility.

\subsection{Empirical credit data}\label{sec:empirical}
We next apply our approach to the
German Credit Dataset~\cite{hofmann1994statlog},
which includes a variety of individual-level socioeconomic and financial characteristics (e.g., age, employment status, and credit history) on a sample of 
\(n = 1{,}000\) people, of whom 700 are deemed creditworthy.
We define the targeted group to be those who currently do not own their residence, a subpopulation that comprises 28\% of the dataset. In this case, 60\% of individuals in the 
targeted group are creditworthy compared to 74\% in the non-targeted group. 
For members of the targeted group, we assume the lender, prior to screening, only knows the targeted group's overall base rate of creditworthiness.
If, however, the lender chooses to screen an applicant in the targeted group, they learn \(\tilde{x}_i\), the applicant's likelihood of being creditworthy conditional on all the available features in the dataset.
For members of the non-targeted group, we assume 
this full estimate of creditworthiness is available prior to screening, and that there is no opportunity to obtain additional information.\footnote{More specifically, at the start of this exercise, we train a logistic regression model on the full dataset predicting creditworthiness as a function of the available covariates.
Then, for members of the targeted group, 
\(\tilde{x}_i\) is the model-estimated probability of creditworthiness for applicant \(i\);
and for the members of the non-targeted group,
that same model estimate is available prior to screening.
}
Finally, we translate estimates of creditworthiness to estimates of utility via Eq.~\eqref{eq:exp-utility}, 
in line with our simulations above.

Figure~\ref{fig:uci} shows the result of applying our 
screening and allocation algorithm to this dataset, 
where we assume the cost of screening \(\cs\) is  \(\$100\) (equal to our high cost regime in the synthetic datasets), the cost of a loan \(\ca\) is  \(\$1{,}000\),
and the total budget is \(B = \$150{,}000\).
Like before, we compare the Pareto frontier of our approach (blue line) to that of
a naive policy in which the lender does not screen applicants (black line).
As with the synthetic datasets above, we
find that the optimal policies with screening substantially outperform those without screening, particularly when we enforce a diversity constraint.
For example, when we require \(\$50{,}000\) of utility to come from allocating loans to the targeted group, the maximum total utility under the no-screening policy is $\$102{,}000$, compared to $\$119{,}000$ under the screening policy, an increase of 17\%.

\begin{figure}[t]
\centering
\includegraphics[width=.99\columnwidth]{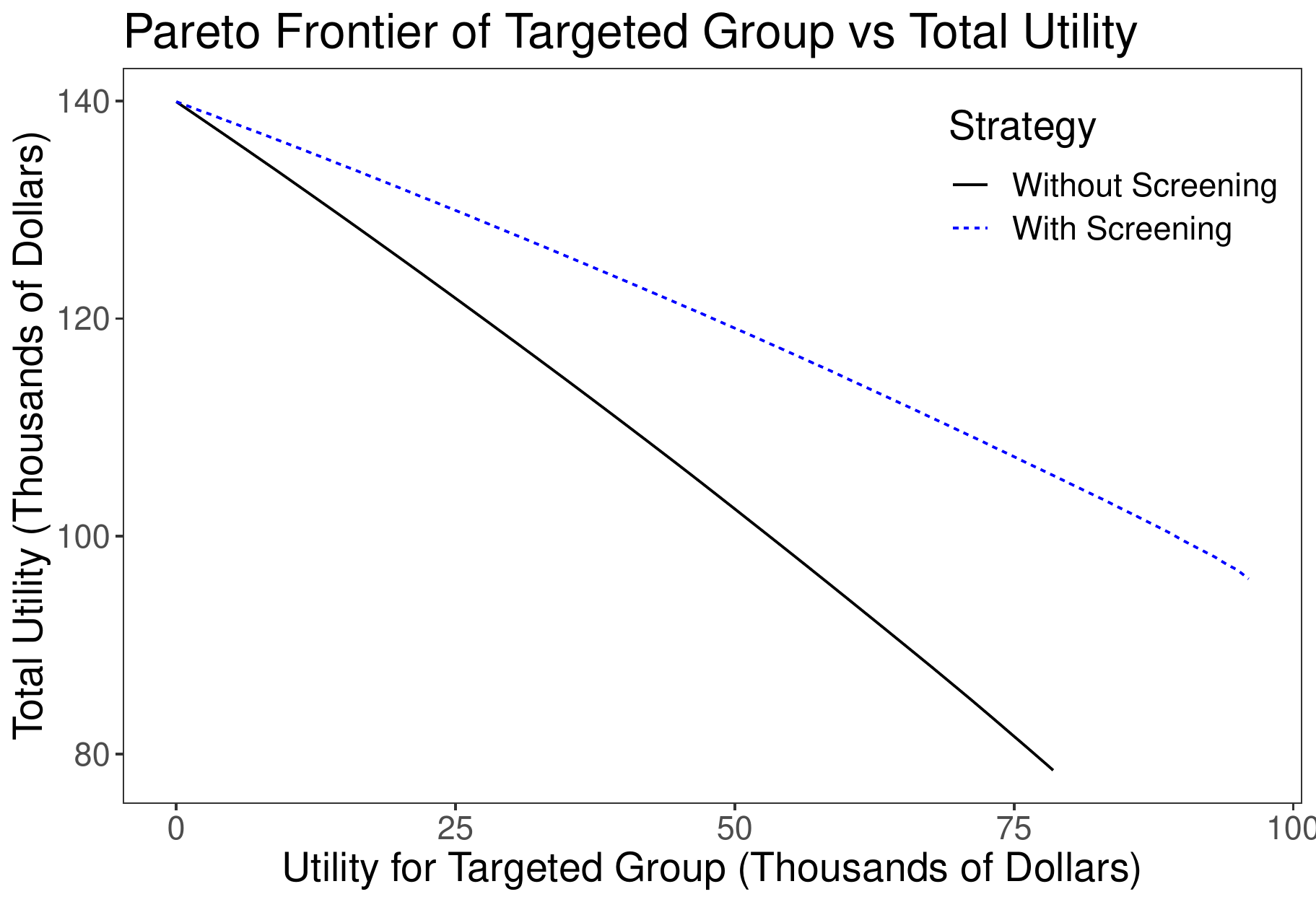}
\caption{For the German Credit Dataset, 
comparison of our optimal screening strategy (blue line) with one without screening (black line).
}
\label{fig:uci}
\end{figure}

\section{Discussion}
Many creditworthy individuals often have difficulty gaining access to traditional credit markets due to lack of formal financial histories,
an issue that can exacerbate existing socioeconomic disparities. 
To address this gap, we developed a simple and efficient algorithm for a budget-constrained decision maker to screen applicants and then allocate a limited resource,
an approach that we find offers substantial benefits on
both real and synthetic lending datasets.
This joint screening-plus-allocation approach is especially useful in settings where a targeted subset of the population---those most in need of an intervention---are also those for whom the least information is available \emph{a priori}, a common situation in many social welfare programs.

Past research has shown that a dearth of high-quality data for various subgroups of the population can lead to poor models in a variety of domains~\cite{gebru2018datasheets,olteanu2019social}, 
including 
text analysis~\cite{bolukbasi2016man,caliskan2017semantics,garg2018word}, 
facial recognition~\cite{buolamwini2018gender}, 
and automated hiring~\cite{dastin2018}. 
Looking forward, our combined data acquisition and decision-making approach provides one
framework to address this challenge by 
jointly modeling the cost of data collection and the value for subsequent improvements in downstream decisions.

\section*{Acknowledgements}
We thank Sam Corbett-Davies, Bernardo Garc{\'i}a Bulle Bueno, and Zhiyuan Jerry Lin for helpful comments and discussions. 
The online appendix is available at: \url{https://arxiv.org/abs/1911.02715}.

\bibliographystyle{ACM-Reference-Format}
\bibliography{sample-base}

\appendix
\section{Proofs}
\label{app:math}

We prove Theorem~\ref{thm:main} in a more general setting in which we allow the cost of screening \(\cs\) and cost of allocation \(\ca\) to vary for each individual.
In the special case of no screening (but with applicant-specific costs of allocation), we note that our optimization problem is equivalent to the fractional knapsack problem. 
For fractional knapsack, the greedy strategy---in which one packs items in descending order of their value per weight---yields an optimal solution~\cite{dantzig1957discrete}.
In our case, we show that the same approach can be used to find an optimal allocation of loans once the set of applicants to screen has been appropriately selected.

To start, we generalize our definition of a threshold policy to account for the applicant-specific allocation costs.

\begin{defn}[Cost-Aware Threshold Policy]
    A \emph{cost-aware threshold policy} is an allocation policy \(\C A\) for some fixed \(t_j \in \B R \cup \{-\infty, \infty\}\) and \(\alpha_j \in [0,1]\), \(j = 1, \ldots, m\), such that
        \begin{equation*}
            A_i = \begin{cases}
                1 & \hat U_i / c_i > t_{g_i},\\
                \alpha_{g_i} & \hat U_i / c_i = t_{g_i}, \\
                0 & \hat U_i / c_i < t_{g_i},
            \end{cases}
        \end{equation*}
    where \(g_i\) denotes the group membership of individual \(i\) and \(c_i\) denotes the cost of allocating resources to individual \(i\).
\end{defn}

We first analyze the setting of a single group, in which case our threshold policy will have a single threshold, \(t\).
Lemma~\ref{lem:main} shows that in this case cost-aware threshold policies are non-dominated:
there cannot be an allocation policy with both higher expected utility and lower cost than a single-threshold policy.
Lemma~\ref{lem:full_range} is an existence theorem for cost-aware threshold policies, which shows that if a particular expected cost or utility can be achieved by an allocation policy, it can be achieved by a threshold policy.

\begin{lem}
\label{lem:main}
    Let \(\C T\) be a cost-aware threshold policy with a single threshold \(t > 0\), and suppose \(\C T\) has an expected cost \(C\)---i.e., \(\sum_i \EE [c_i \cdot T_i] = C\)---and expected utility \(\Upsilon\)---i.e., \(\sum_i \EE [U_i \cdot T_i] = \Upsilon\). Let \(\C A\) be any other allocation policy.
    \begin{description}
        \item[Case 1:] If the expected cost of \(\C A\) is \(C\), then the expected utility of \(\C A\) is less than or equal to \(\Upsilon\).
        \item[Case 2:] If the expected utility of \(\C A\) is \(\Upsilon\), then the expected cost of \(\C A\) is greater than or equal to \(C\).
    \end{description}
\end{lem}

\begin{proof}
Let \(\Delta_i = T_i - A_i\), and set \(\Delta^-_i\) be the negative part of \(\Delta_i\), i.e., \(\Delta_i \cdot \B 1_{\Delta_i \leq 0}\).
Similarly, let \(\Delta^+_i\) be the positive part.

The variables \(\Delta^+_i\) represent circumstances in which \(\C A\) ``saves money'' in expectation over \(\C T\), and \(\Delta_i^-\) represents how this portion of the budget is reallocated.

Consider the following events:
    \begin{align*}
        E^>_i &= \{\hat U_i > c_i \cdot t\},\\
        E^=_i &= \{\hat U_i = c_i \cdot t\},\\
        E^<_i &= \{\hat U_i < c_i \cdot t\}.
    \end{align*}
Note that \(\Delta^+_i \neq 0\) only on some subset of \(E^>_i \cup E^=_i\).
likewise, \(\Delta^-_i \neq 0\) only on a subset of \(E^<_i \cup E^=_i\).
Informally, \(\C A\) can only save in expected cost over \(\C T\) when \(\hat U_i \geq t\), and can only reallocate savings to circumstances in which \(\hat U_i \leq t\).

Because \(\Delta^+_i \geq 0\) and \(\Delta^-_i \leq 0\), it consequently follows that
    \begin{equation}
    \label{eq:gt}
        \Delta^+_i \cdot \hat U_i \geq \Delta^+_i \cdot c_i \cdot t,
    \end{equation}
and
    \begin{equation}
    \label{eq:lt}
        \Delta^-_i \cdot \hat U_i \geq \Delta^-_i \cdot c_i \cdot t.
    \end{equation}
In particular,
    \begin{align*}
        \EE \left[ \sum_{i = 1}^n \Delta_i \hat U_i \right]
            &= \sum_{i = 1}^n \EE [\Delta_i \hat U_i]\\
            &= \sum_{i = 1}^n \bigg( \Pr(E_i^>) \cdot \EE [\Delta_i^+ \cdot \hat U_i \mid E_i^>]\\
            &\hspace{0.9cm} + \Pr(E_i^=) \cdot \EE [\Delta_i \cdot \hat U_i \mid E_i^=]\\
            &\hspace{0.9cm} + \Pr(E_i^<) \cdot \EE [\Delta_i^- \cdot \hat U_i \mid E_i^<] \bigg)\\
            &\geq \sum_{i = 1}^n \bigg( \Pr(E_i^>) \cdot \EE [\Delta_i^+ \cdot c_i \cdot t \mid E_i^>]\\
            &\hspace{0.9cm} + \Pr(E_i^=) \cdot \EE [\Delta_i \cdot c_i \cdot t \mid E_i^=]\\
            &\hspace{0.9cm} + \Pr(E_i^<) \cdot \EE [\Delta_i^- \cdot c_i \cdot t \mid E_i^<] \bigg)\\
            &= t \cdot \EE \left[ \sum_{i = 1}^n c_i \Delta_i \right].
    \end{align*}
Here the inequality follows from Eqs.~\eqref{eq:gt} and \eqref{eq:lt}.

Now, since any allocation policy is by definition conditionally independent of \(U_i\) given \(\hat U_i\), it follows that \(\EE [A_i U_i] = \EE [A_i \hat U_i]\) and \(\EE [T_i U_i] = \EE [T_i \hat U_i]\). In consequence,
\begin{equation}
    \label{eq:key}
        \EE \left[ \sum_{i = 1}^n \Delta_i U_i \right] \geq t \cdot \EE \left[ \sum_{i = 1}^n c_i \Delta_i \right].
    \end{equation}

In Case 1, the expected costs of \(\C A\) and \(\C T\) are equal. It follows that \(\EE \left[ \sum_i c_i \Delta_i \right] = 0\); consequently, by Eq.~\eqref{eq:key}, \(\EE \left[ \sum_i \Delta_i U_i \right] \geq 0\), and so the expected utility of \(\C T\) is greater than or equal to that of \(\C A\).

In Case 2, the expected utilities of \(\C A\) and \(\C T\) are equal, so that \(\EE \left[ \sum_i \Delta_i U_i \right] = 0\). Then, again by Eq.~\eqref{eq:key}, since \(t > 0\), \(\EE \left[ \sum_i c_i \Delta_i \right] \leq 0\). Therefore the expected cost of \(\C T\) is less than or equal to that of \(\C A\).
\end{proof}

\begin{lem}
\label{lem:full_range}
    Let \(\Upsilon > 0\) be an achievable expected utility---that is, there is an allocation policy \(\C A\) such that \(\EE[\sum_i A_i U_i] = \Upsilon\)---and let \(C\) be an achievable expected cost.
    Then,
        \begin{itemize}
            \item There exists a cost-aware threshold policy \(\C T\) of expected utility \(\Upsilon\); and,
            \item There exists a (possibly different) cost-aware threshold policy \(\C T\) of expected cost \(C\).
        \end{itemize}
\end{lem}

\begin{proof}
    We first address the case of a given positive expected utility.
    Consider the function \(\Upsilon(t) : [0, \infty) \to \B R\) given by
        \begin{equation}
            \Upsilon(t) = \sum_{i = 1}^n \EE \left[ U_i \cdot \B 1_{\hat U_i / c_i \geq t} \right].
        \end{equation}
    We note three facts about \(\Upsilon(t)\).
    First, \(\Upsilon(t)\) is monotonically non-increasing for \(t \geq 0\).
    Second, \(\Upsilon(0)\) is the maximum expected utility achievable by any allocation policy.
    Third, \(\Upsilon(t) \to 0\) as \(t \to \infty\).
    
    Suppose \(0 < \upsilon \leq \Upsilon(0)\).
    Then there exists some maximal \(t\) such that \(\Upsilon(t) \geq \upsilon\).
    Let \(\rho = \sum_i \Pr(\hat U_i / c_i = t)\), \(\tau = \sum_i c_i \cdot t_i\), and \(\delta = \Upsilon(t) - \upsilon\).
    Define \(\alpha = \delta / (\tau \cdot \rho)\).
    Then the threshold policy defined by \((t, \alpha)\) will have expected utility \(\upsilon\).
    The result follows as all achievable positive expected utilities lie in the half-open interval \((0, \Upsilon(0)]\).
    
    The proof in the case of expected cost is similar.
\end{proof}

We are now equipped to move to analyze settings in which there are multiple groups. We prove the following generalization of Theorem~\ref{thm:main}.

\begin{thm}
     Suppose the constrained optimization problem defined by Eqs.~\eqref{eq:def1}, \eqref{eq:def2}, and \eqref{eq:def3} has a solution \((\C P^*, \C A^*)\).
     Then there is a (not necessarily single-threshold) cost-aware threshold policy \(\C T^*\) such that \((\C P^*, \C T^*)\) is also a solution.
\end{thm}

\begin{proof}[Proof of Theorem \ref{thm:main}]
    Consider a solution screening and allocation policy pair \((\C P^*, \C A^*)\) satisfying the allocation diversity constraints.
    
    Note that each of the collections \(\C A_j^* = \{A_i^*\}_{i \in G_j}\) defines an allocation policy on the subpopulation \(G_j\).
    Applying Lemma~\ref{lem:full_range} to each of them yields a threshold policy \(\C T_j^*\)---defined by the pair \((t_j, \alpha_j)\)---of the same expected cost as \(\C A_j^*\). 
    By Lemma~\ref{lem:main}, \(\C T_j^*\) achieves the same or higher expected utility on \(G_j\) as \(\C A_j^*\).
    
    Let \(\C T^*\) be the threshold policy defined by thresholds \(t_1, \ldots, t_m\) and boundary randomization probabilities \(\alpha_1, \ldots, \alpha_m\).
    The expected cost of \(\C T^*\) is the same as \(\C A^*\).
    Moreover, \(\C T^*\) satisfies the group utility constraints, since the utility it achieves on each group \(G_j\) is greater than or equal to \(\C A^*\).
    Lastly, the expected utility achieved by \(\C T^*\) is greater than or equal to that of \(\C A^*\).
    
    Since \(\C A^*\) was assumed optimal, it must be that the expected utilities are equal. Therefore the pair \((\C P^*, \C T^*)\) is also a solution to the constrained optimization problem defined by Eqs.~\eqref{eq:def1}, \eqref{eq:def2}, and \eqref{eq:def3}.
\end{proof}

Theorem~\ref{thm:main} follows as an immediate corollary.

In some cases, it is useful to impose equality rather than inequality diversity allocation constraints.
For instance, in some situations the decision maker may wish to obtain the Pareto frontier.
Lemma~\ref{lem:modest_gen} ensures that it is still sufficient when specific utilities must be achieved on subgroups to restrict to the set of threshold policies when searching for a globally optimal screening and allocation policy.

\begin{lem}
\label{lem:modest_gen}
    The conclusion of Theorem~\ref{thm:main} still holds even if the inequalities in Eq.~\eqref{eq:def3} are replaced with equalities, i.e., if the optimal policy pair \((\C P^*, \C A^*)\) satisfies
        \begin{equation}
            \EE \left[ \sum_{i \in G_j} U_i A_i \right] = \Lambda_j.
        \end{equation}
\end{lem}

\begin{proof}
    The proof is identical to that of Theorem~\ref{thm:main}, except that Lemma~\ref{lem:full_range} is used to construct a threshold policy \(\C T^*\) achieving exactly utility \(\Lambda_j\) on any constrained group \(G_j\).
    Case 2 of Lemma~\ref{lem:main} then shows that this utility is achieved at possibly less expected cost than \(\C A^*\).
    It follows immediately that the pair \((\C P^*, \C T^*)\) is a solution to the optimization problem.\footnote{
        In this case, it is actually possible that \(\C T^*\) will result in expected cost savings over \(\C A^*\), even though \(\C A^*\) is optimal. This can occur if there is no remaining positive utility to ``spend'' the savings on, since resources are already allocated at total expected cost less than \(B\) in all cases where \(\hat U_i\) is positive.
    }
\end{proof}

\section{Optimized Procedure}
\label{app:lp}

In our experiments we consider cases in which screening provides no additional information about one group of individuals (i.e., \(D_i\) is concentrated at a single point for \(i \in G_j\)).
In this case, it is possible to significantly reduce the number of LPs the decision maker must solve to find an optimal policy or calculate the Pareto frontier.
For instance, a lender may already possess all available information relevant to the creditworthiness of those with traditional credit histories.

In this case, the cost of allocating resources to group \(G_j\) does not depend on the screening policy \(\C P\).
Therefore, rather than sweeping over all possible pairs of thresholds and boundary randomization probabilities \((t_j, \alpha_j)\), 
one can encode the allocation policy to \(G_j\) directly as a collection of allocation probabilities \(\{a_i\}_{i \in G_j}\).

For notational simplicity, suppose without loss of generality that \(m = 2\), that \(G_1 = \{1, \ldots, n_1\}\) and \(G_2 = \{n_1 + 1, \ldots, n_1 + n_2\}\), and that \(D_i\) is a pointmass for all \(i \in G_2\).

Then, let \(p = (p_1, \ldots, p_{n_1})\) be a length \(n_1\) vector of probabilities indicating the probability of screening the members of \(G_1\). Let \(a = (a_{n_1 + 1}, \ldots, a_{n_1 + n_2})\) be a length \(n_2\) vector of probabilities indicating the probability of allocating resources to members of \(G_2\). We keep the rest of the notation the same as in Section \ref{sec:solution}.

Fix threshold \(t_1\) and boundary randomization probability \(\alpha_1\) for \(G_1\). Then, the expected utility of any pair of screening policies (for \(G_1\)) and allocation policies (for \(G_2\)) can be expressed as a linear function of \(p\) and \(a\):
    \begin{equation}
    \label{eq:lp1}
        \sum_{i \in G_1} \left[ q_i \cdot e_i \cdot p_i + o_i \cdot \mu_i \cdot (1 - p_i) \right] + \sum_{i \in G_2} \left[ \mu_i \cdot a_i \right].
    \end{equation}

Likewise, our budget constraint can be expressed as follows:
    \begin{equation}
    \label{eq:lp2}
        \begin{split}
            B &\geq \sum_{i \in G_1} \left[ \cs \cdot p_i + \ca \cdot q_i \cdot p_i + \ca \cdot o_i \cdot (1-p_i) \right] \\
                &\hspace{.5cm} + \sum_{i \in G_2} \left[ \ca \cdot a_i \right]
        \end{split}
    \end{equation}
The allocation diversity constraint will take one of the following two forms. If the constrained group is \(G_1\),
    \begin{equation}
    \label{eq:lp3}
        \sum_{i \in G_1} \left[ q_i \cdot e_i \cdot p_i + o_i \cdot \mu_i \cdot (1 - p_i) \right] = \Lambda_1.
    \end{equation}
Otherwise, if the constrained group is \(G_2\), the constraint will be
    \begin{equation}
    \label{eq:lp4}
        \sum_{i \in G_2} \left[ \mu_i \cdot a_i \right] = \Lambda_2.
    \end{equation}
Lastly, we ensure that \(a\) and \(p\) are probability vectors:
    \begin{align}
    \label{eq:lp5}
        a_i &\leq 1,\\
        p_i &\leq 1.
    \end{align}

Together Eqs.~\eqref{eq:lp1}, \eqref{eq:lp2}, \eqref{eq:lp3}, and \eqref{eq:lp4}---given the initial data of the threshold policy \(\C T_1\) on \(G_1\) defined by the pair \((t_1, \alpha_1)\)---define a linear program in the \(n_1 + n_2\) decision variables \(p = (p_1, \ldots, p_n)\) and \(a = (a_{n_1 + 1}, \ldots, a_{n_1 + n_2})\) that can be solved for an optimal screening policy \(\C P^*\) (on \(G_1\)) and allocation policy \(\C T_2^*\) (on \(G_2\)).
By the results of Appendix \ref{app:math}, \(\C T_2^*\) can be assumed to have the form of a threshold policy.
Sweeping over all such pairs in a (discretized) space \(\B R \times [0, 1]\), the resulting triple \((\C P^*, \C T_1^*, \C T_2^*)\) that maximizes utility will represent a globally optimal screening and allocation policy.

\end{document}